\documentclass[journal]{IEEEtran}
\ifCLASSINFOpdf
\else
\fi

\usepackage{cite}
\usepackage{amsmath,amssymb,amsfonts}
\usepackage{algorithmic}
\usepackage{graphicx}
\usepackage{textcomp}
\usepackage{xcolor}
\usepackage{amsthm}
\usepackage{balance} 
\usepackage{tikz} 
\usetikzlibrary{positioning} 
\usetikzlibrary{calc,shapes} 
\usepackage{tikz-3dplot}
\usepackage{clipboard}
\usepackage{xr}
\newtheorem{theorem}{Theorem}[section]
\newtheorem{lemma}[theorem]{Lemma}

\newtheorem{remark}[theorem]{Remark}

\begin{document}
%
\title{Available Degrees of Spatial Multiplexing of a Uniform Linear Array with Multiple Polarizations: A Holographic Perspective}
%
%
%

\author{Xavier~Mestre,~\IEEEmembership{Senior Member,~IEEE,}
        Adrian~Agustin,~\IEEEmembership{Senior Member,~IEEE,}
        David~Sardà 
\thanks{This work is supported by the grant from the Spanish ministry of economic affairs and digital transformation and of the European union - NextGenerationEU UNICO-5G I+D/AROMA3D-Earth (TSI-063000-2021-69), by Grant 2021 SGR 00772 funded by the Universities and Research Department from Generalitat de Catalunya, and by the Spanish Government through the project 6G AI-native Air Interface (6G-AINA, PID2021-128373OB-I00 funded by MCIN/AEI/10.13039/501100011033) ERDF A way of making Europe. }
}


%
%

\markboth{Submitted to IEEE Open Journal on Signal Processing}{ICASSP 2025}
%



\maketitle

\begin{abstract}
The capabilities of multi-antenna technology have recently been significantly enhanced by the proliferation of extra large array architectures. The high dimensionality of these systems implies that communications take place in the near-field regime, which poses some questions as to their effective performance even under simple line of sight configurations. In order to study these limitations, a uniform linear array (ULA) is considered here, the elements of which are three infinitesimal dipoles transmitting different signals in the three spatial dimensions. The receiver consists of a single element with three orthogonal infinitesimal dipoles and full channel state information is assumed to be available at both ends. A capacity analysis is presented when the number of elements of the ULA increases without bound while the interelement distance converges to zero, so that the total aperture length is kept asymptotically fixed. In particular, the total number of available spatial degrees of freedom is shown to depend crucially on the receiver position in space, and closed form expressions are provided for the different achievability regions. From the analysis it can be concluded that the use of three orthogonal polarizations at the transmitter guarantees the universal availability of at least two spatial streams everywhere.
\end{abstract}


%
\IEEEpeerreviewmaketitle

\section{Introduction}

Current wireless communication systems strongly rely on the use of multi-antenna technology to push the performance limits in terms of reliability and throughput. As the frequency of operation of these systems moves to higher, more available, bands e.g., millimeter wave (mmWave) and sub-terahertz (THz) frequencies \cite{Priebe14,Shafi18}, the use of multiantenna technology has become necessary element to overcome the more adverse propagation conditions. Given the stronger propagation loss and the higher material absorption at these frequency bands, wireless communication channels tend to  present a strong line-of-sight component. 
In parallel with this trend, massive antenna deployments and extra large antenna arrays are progressively being introduced in order to meet the strict performance requirements of future wireless systems. The large size of these antenna deployments is becoming comparable to the distance between transmitter and receiver, which is in turn decreasing in order to compensate for the high propagation loss. This is a radical change with respect to conventional wireless architectures, which are typically assumed to work in the far field. In these new scenarios, the electromagnetic wave fronts impinging on the multi-antenna system can no longer be assumed to be planar, and more complicated near field channel models are instead applicable. This may have some additional performance benefits in the line of sight regime, since the near-field structure of the channel becomes geometrically similar to a conventional multi-path configuration. In fact, contrary to what happens in the far-field regime where line of sight MIMO channels are essentially low rank, in the near-field a MIMO channel matrix quickly becomes full rank thanks to the different electric distances with the various transmit elements \cite{Driessen99,Jiang05,Bohagen09,Sanguinetti23}. The purpose of this paper is to explore the spatial multiplexing capabilities that can be exploited in this near-field configuration. 

There is a very rich body of recent literature studying the capabilities of multi-antenna technologies in the near-field. For example, the work in \cite{Larsson05,Torkildson2011} demonstrate how parallel uniform linear arrays (ULAs) with an appropriate inter-element separation distance can generate a MIMO channel matrix with identical singular values, which is optimal for spatial multiplexing. 
More general orientations were studied in \cite{Gesbert02}, which evaluated the sensitivity of the rank of the line-of-sight MIMO channel matrix and established how full channel orthogonality conditions (which are the best for spatial multiplexing) can be achieved in short transmission distances. Other analyses considered different geometric optimizations \cite{Bohagen07,Sarris07} under the assumption that the transmission distance is much larger than the physical size of the transmit and receive arrays.
Several other works have focused on the optimization of the geometry and orientation of the ULAs in order to achieve optimal performance, taking into account a wide range of values of the signal to noise ratio (SNR) \cite{Do21}.
More recent approaches have focused on refining the different approximations in the channel model \cite{RuizSicilia24} that is used to optimize the transmit and receive array configurations.

In parallel with all this, multiple studies have additionally considered the use of holographic surfaces \cite{Dardari21,Huang20,Gong24survey,gong2024near,Pizzo22}, which are also referred to as continuous aperture architectures \cite{Sayeed10,Sayeed11} or large intelligent surfaces \cite{Hu18} in the array processing literature. These array configurations assume an asymptotically large number of antennas that are separated by an asymptotically small inter-element distance. 
Holographic surfaces function as a continuous antenna array that can dynamically shape and control the signal generation and its electromagnetic projection, creating ``holograms" of electromagnetic fields (see \cite{Gong24survey} for an extensive review of hardware architectures and design methodologies of holographic surfaces). 

Most of the previous literature has focused on the optimization of both signal processing procedures and antenna array configurations by considering an abstract model of the radiating pattern of both transmit and receive antennas. One of the few recent exceptions is the work in \cite{Wei23}, which considers a multi-user MIMO setting based on a channel model with three spatial polarizations based on the Green dyadic function. Building on this channel model, the authors introduced two precoding schemes aimed at mitigating cross-polarization and inter-user interference. This channel model has also been used in \cite{gong2024near} to characterize the capacity of arbitrary placed holographic surfaces made of densely-packed surface radiators. A slightly simpler channel model was recently adopted in \cite{agustin23}, where we considered infinitesimal dipoles in three different orthogonal polarizations at both sides of the communication link in a holographic set-up (high number of elements, short inter-element distance). 

In fact, the discrete dipole formalism has been widely used to characterize metasurface antennas \cite{Imani20}, in some cases by directly modeling the interaction among the different dipoles \cite{Draine94,Landy14,Johnson14,Pulido17,Yoo19}, and in some cases by simply omitting these coupling effects altogether \cite{Lipworth13,Lipworth15,Smith17,Imani18}. 
The discrete dipole approximation is a very simple abstraction of a holographic configuration where infinitesimal dipoles can be seen as the atomic components that may ultimately conform these radiating elements. It was recently shown in \cite{agustin24twc} that the behavior of both ULAs and uniform planar arrays (UPAs) can essentially be described by the relative geometry of the scenario, so that one can optimize the size of the radiating surface to ensure an optimum transmission rate at a certain point in space. 

The present paper strongly builds upon the asymptotic results in \cite{agustin24twc} by further investigating the analytical expressions for the three-polarized channel matrices in the holographic regime. More specifically, the expressions derived in \cite{agustin24twc} are employed here to establish the multiplexing capabilities of a ULA as a function of the SNR. We provide closed form analytical expressions for the coverage areas where a holographic array is able to support one, two or three spatial streams depending on the number of polarizations that are used at the transmitter side. Results allow to conclude that the use of three orthogonal polarizations guarantees the universal availability of two spatial streams everywhere.

\section{Scenario and Signal Model}

Let us consider a scenario where we have a transmitting uniform linear array consisting of $2M+1$ elements ($M>0$) separated by a distance of $\Delta_T$ meters, each one of them equipped with three orthogonal infinitesimal dipoles, see further Fig.~\ref{fig:Scenario}. 
The receiver employs a single spatial element, also consisting of three orthogonal dipoles. We assume that the three dipoles at each ULA element are fed with different radio frequency (RF) chains, so that the transmitter digitally generates up to $3(2M+1)$ signals to be transmitted.
Without loss of generality, we will assume that the three dipoles at each element of the transmit array are aligned with the coordinate system. 
Let $\mathbf{H}_m$ denote the $3 \times 3$ channel response between the $m$th element of the ULA ($m=-M,\ldots,M$) and the receiver. Assuming a narrowband transmission at a wavelength $\lambda$ and disregarding the reactive terms (which decay quite fast with the distance), the electromagnetic channel takes the form \cite{Poon05} 
\begin{equation} \label{eq:channelModel}
    \mathbf{H}_m = \frac{\xi}{ \lambda \Vert \mathbf{r}_m \Vert} \exp \left(-\mathrm{j} \frac{2 \pi}{\lambda} \Vert \mathbf{r}_m \Vert\right) \mathcal{Q} \left[ \mathbf{I}_3 - \frac{\mathbf{r}_m\mathbf{r}_m^H}{\Vert\mathbf{r}_m \Vert^2}\right] 
\end{equation}
where $\xi$ is a certain complex constant\footnote{If $\xi$ is chosen as $\xi = {\mathrm{j}\eta}/{2}$ where $\eta$ denotes the permittivity of the medium, then the three columns of $\mathbf{H}_m$ may be identified as the electric field generated by each of the three orthogonal dipoles at the $m$th transmitter.}, $\mathbf{r}_m$ is the vector between the $m$th transmitting element of the ULA and the receiver and $\mathcal{Q}$ is a rotation matrix that accounts for the fact that the dipoles at the receiver may not be aligned with those at the transmitter. It can be seen that this rotation matrix does not really affect the capacity of the channel when the receiver has access to the three polarizations. The corresponding rotation is essentially superfluous to the discussion in this work, so that we will assume from now on that $\mathcal{Q} = \mathbf{I}_3$.

We will assume for simplicity that the receiver is located on the $yz$-plane, so that the vector $\mathbf{r}_m$ takes the form
\begin{equation}
\mathbf{r}_m = \left[0,D \sin \theta - m \Delta_T , D \cos \theta \right]^T
\end{equation}
where $D>0$ is the distance between the center of the array and the receiver and where $\theta \in (-\pi/2,\pi/2)$ is the elevation angle (see further Fig.~\ref{fig:Scenario}). 
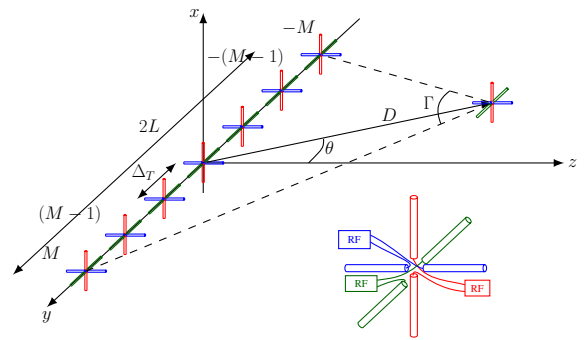
\begin{figure}[h] 
    \centering 
\centerline{

\begin{tikzpicture}[>={latex},scale=0.4,every node/.style={scale=0.4}]

\def\tripole at (#1,#2){
\node (A) [scale = 0.3, cylinder, shape border rotate=90, shape aspect=.5, draw = red, minimum height= 60, minimum width=1] at (#1+0,#2+0.35) {};
\node (A) [scale = 0.3, cylinder, shape border rotate=90, shape aspect=.5, draw = red, minimum height= 60, minimum width=1] at (#1+0,#2-0.35) {};
\node (A) [scale = 0.3,cylinder, shape border rotate=0, shape aspect=.5, draw = blue, minimum height= 60, minimum width=1] at (#1+0.35,#2+0) {};
\node (A) [scale = 0.3,cylinder, shape border rotate=0, shape aspect=.5, draw = blue, minimum height= 60, minimum width=1] at (#1-0.35,#2+0) {};
\node (A) [scale = 0.3,cylinder, rotate=43, shape aspect=.5, draw = black!60!green, minimum height= 60, minimum width=1] at (#1+0.35/1.25,#2+0.32/1.25) {};
\node (A) [scale = 0.3,cylinder, rotate=43, shape aspect=.5, draw = black!60!green, minimum height= 60, minimum width=1] at (#1-0.35/1.2,#2-0.33/1.2) {};
}

\def\tripoleGros at (#1,#2){
\node (A) [cylinder, shape border rotate=90, shape aspect=.5, draw = red, minimum height= 60, minimum width=1] at (#1+0,#2+1.3) {};
\node (A) [cylinder, shape border rotate=90, shape aspect=.5, draw = red, minimum height= 60, minimum width=1] at (#1+0,#2-1.3) {};
\node (A) [cylinder, shape border rotate=0, shape aspect=.5, draw = blue, minimum height= 60, minimum width=1] at (#1+1.3,#2+0) {};
\node (A) [cylinder, shape border rotate=0, shape aspect=.5, draw = blue, minimum height= 60, minimum width=1] at (#1-1.3,#2+0) {};
\node (A) [cylinder, rotate=43, shape aspect=.5, draw = black!60!green, minimum height= 60, minimum width=1] at (#1+1.3/1.25,#2+1.1/1.25) {};
\node (A) [cylinder, rotate=43, shape aspect=.5, draw = black!60!green, minimum height= 60, minimum width=1] at (#1-1.3/1.2,#2-1.5/1.2) {};
}

\draw [->] (0,0) node (v4) {} -- (12,0) node[right] {\LARGE $z$};
\draw [->] (0,-1) -- (0,5) node[left] {\LARGE $x$};
\draw [<-] (-1.3*4,-1.2*4)  node[below] {\LARGE $y$} -- (1.3*4,1.2*4) ;
\draw [<->] (-1.3*3.1-2.3,-1.2*3.1)  -- (1.3*3.1-2.3,1.2*3.1) ;
\node [right] at (-2.25,1.25) {\LARGE $2L$};

 \foreach \x in {-3,...,3}
       	\tripole at (1.3*\x,1.2*\x);


%


\draw [<->] (-2.2,-1.2) -- (-0.9,0);
\node [right] at (-2.5,-0.3) {\LARGE $\Delta_T$};

\node [right] at (-5.5,-3) {\LARGE $M$};
\node [right] at (-5.6,-1.7) {\LARGE $(M-1)$};
\node [right] at (0.0,3.5) {\LARGE $-(M-1)$};
\node [right] at (2.5,4.5) {\LARGE $-M$};
 
 
\node [right] at (5.8,1.6) {\LARGE $D$};

\tripoleGros at (7, -3.5); 

\begin{scope} [scale=1.5, shift={(-5.5,1.5)}]
\draw[red] (10.15,-3.62) .. controls (10.35,-4.1) and (10.9,-4.1) .. (11.3,-4.2);
\draw[red] (10.15,-4) .. controls (10.15,-3.7) and (10.5,-4.3) .. (11.3,-4.35);
\draw [blue] (9.99,-3.83) .. controls (10.34,-3.74) and (9.65,-3.54) .. (9.2,-3.25);
\draw [blue] (10.37,-3.84) .. controls (10.14,-3.81) and (9.65,-3.3) .. (9.2,-3.1);
\draw [black!60!green] (10.36,-3.68) .. controls (9.95,-4.05) and (9.7,-4.05) .. (9.35,-4.1);
\draw [black!60!green] (9.95,-4.15) .. controls (10.05,-4.05) and (9.8,-4.05) .. (9.35,-4.25);
\draw [blue] (9.2,-3) rectangle node {RF} (8.5,-3.4);
\draw [red] (11.85,-4.1) rectangle node {RF} (11.3,-4.45);
\draw [black!60!green] (8.65,-4) rectangle node {RF} (9.35,-4.35);

\end{scope}

\draw[->] (v4) -- (9.6,2) node (v1) {};
 \tripole at (9.6,2);

\draw (3.5,0) .. controls (4,0.5) and (4,0.5) .. (4,0.8);
\node at (4.2,0.5) {\LARGE $\theta$};

\draw[dashed] (v1) -- (3.9,3.6);
\draw[dashed] (-3.9,-3.6) -- (v1) -- cycle;
\draw (7.9,1.3) .. controls (7.7,1.8) and (7.8,2.1) .. (8.2,2.4);

\node at (7.5,2) {\LARGE $\Gamma$};
\end{tikzpicture}}
\caption{Scenario configuration. The transmitter is a ULA consisting of $2M+1$ elements, each incorporating 3 orthogonal infinitesimal dipoles. The receiver is assumed to lie on the $yz$-plane. Here, $D$ is the distance between the receiver and the center of the array, $\theta$ is the elevation angle and $\Gamma$ is the angle of view of the ULA from the receiver.}
    \label{fig:Scenario} 
\end{figure}

In some parts of the paper we will study the case where the transmit array only employs the two polarizations along the $xy$-plane. In this case, the corresponding channel can be modeled by selecting the two first columns of the matrix in (\ref{eq:channelModel}). More generally, we will denote by $t_\mathrm{pol}$ and $r_\mathrm{pol}$ the number of polarizations that are employed at the transmitter and receiver respectively, and we will denote the corresponding channel as
$$
\mathbf{H}_m^{t_\mathrm{pol} \times r_\mathrm{pol}} = \left[\mathbf{H}_m\right]_{1:r_\mathrm{pol},1:t_\mathrm{pol}}
$$
which corresponds to the selection of the first $r_\mathrm{pol}$ rows and $t_\mathrm{pol}$ columns of $\mathbf{H}_m$ respectively. 
\begin{remark}
The fact that transmitter and receiver are both located on the $yz$-plane does not imply any loss of generality in the fully polarized case (whereby $t_\mathrm{pol}=r_\mathrm{pol} = 3$) since the achievable rate is independent of any rotation along the $y$-axis. However, this is not the case when less polarizations are used, e.g. when $t_\mathrm{pol}=2$. In this second case, the study presented here can easily be generalized to any general azimuth by particularizing the more general expressions derived in \cite{agustin24twc} for the uniform planar array.    
\end{remark}

Let $\mathbf{H}^{t_\mathrm{pol} \times r_\mathrm{pol}}$ denote the $r_\mathrm{pol} \times (2M+1)t_\mathrm{pol}$ global channel matrix, formed by stacking the matrices $\mathbf{H}_m^{t_\mathrm{pol} \times r_\mathrm{pol}}$ side by side, that is
\[
\mathbf{H}^{t_\mathrm{pol} \times r_\mathrm{pol}} 
= \left[\mathbf{H}_{-M}^{t_\mathrm{pol} \times r_\mathrm{pol}} ,\ldots,\mathbf{H}_{M}^{t_\mathrm{pol} \times r_\mathrm{pol}} \right].
\]
The signal at the $r_\mathrm{pol}$ outputs of the receiver can be expressed as 
\[
\mathbf{y} = \mathbf{H}^{t_\mathrm{pol} \times r_\mathrm{pol}} \mathbf{x} + \mathbf{n}
\]
 where $\mathbf{x}\in \mathbb{C}^{(2M+1)t_\mathrm{pol} \times 1}$ contains the transmitted signals and where $\mathbf{n}\in \mathbb{C}^{r_\mathrm{pol} \times 1}$ is the background noise, assumed to be complex circularly symmetric Gaussian distributed with power $\sigma^2$, i.e. $\mathbf{n} \sim  \mathcal{CN}\left(\mathbf{0},\sigma^2\mathbf{I}_{r_\mathrm{pol}}\right)$. The above signal model can be seen as an approximation of the general solution to the Helmholtz wave equation, which is expressed as the integral of the dyadic Green function times the electric current at the source \cite{Poon05}. Changing the original integral of the exact solution by the corresponding Riemann sum, we directly obtain the signal model above. When $M\rightarrow \infty$ while $\Delta_T \rightarrow 0$ the sum converges to the Riemann integral, which essentially describes the general solution to the wave equation under continuous spatial excitation. 



Assume that the input signal $\mathbf{x}$ is random, zero mean and has covariance matrix $\mathbf{Q} = \mathbb{E}\left[ \mathbf{xx}^H \right] $. It is well known \cite{Tse05,Palomar05} that if the transmitter has access to the channel matrix $\mathbf{H}^{t_\mathrm{pol} \times r_\mathrm{pol}}$, the optimum transmission rate is obtained under Gaussian signaling, i.e. $\mathbf{x}\sim \mathcal{CN}(0,\mathbf{Q})$, and is given by
\[
C = \sup_{\substack{\mathbf{Q} \geq 0 \\ \mathrm{tr}\mathbf{Q}=P}}  \log \det \left( \mathbf{I}_{r_\mathrm{pol}} + \mathbf{H}^{t_\mathrm{pol} \times r_\mathrm{pol}} \mathbf{Q} \left(\mathbf{H}^{t_\mathrm{pol} \times r_\mathrm{pol}}\right)^H / \sigma^2  \right)
\]
where $P$ is the total transmitted power. It well known that the optimum
covariance $\mathbf{Q}$ is obtained as $\mathbf{Q} = \mathbf{V} \mathbf{P}\mathbf{V} ^H$ where $\mathbf{V}$ is a $(2M+1)t_\mathrm{pol} \times r_\mathrm{pol}$ matrix obtained as the left singular vectors associated with positive singular values of the channel matrix $\mathbf{H}^{t_\mathrm{pol} \times r_\mathrm{pol}} $ and where $\mathbf{P}$  is an $ r_\mathrm{pol} \times  r_\mathrm{pol}$ diagonal matrix with non-negative elements denoted as $p_1,\ldots,p_{r_\mathrm{pol}}$ obtained from the water-filling equation. 

All this implies that the optimum transmitter can be built as $\mathbf{x} = \mathbf{V}\mathbf{P}^{1/2} \mathbf{s}$, where $\mathbf{s} \sim \mathcal{CN}(0,\mathbf{I})$ contains the i.i.d. signaling symbols that encode the transmitted message. Therefore, we may identify the number of positive diagonal entries of $\mathbf{P}$ as the maximum number of parallel transmissions (spatial degrees of freedom) that can be multiplexed in a certain scenario. Depending on the SNR and the position of the receiver, the number positive entries may oscillate between $1$ (single beamforming) and $r_\mathrm{pol}$ (maximum spatial multiplexing capabilities). The objective of this paper is to characterize the regions where these are achieved. 

\section{Holographic regime}

In order to study the behavior of the above system, we consider here a holographic approximation of the setting, whereby the number of elements of the transmit ULA is assumed to be asymptotically large ($M\rightarrow\infty$) whereas the distance between consecutive elements converges to zero ($\Delta_T\rightarrow 0$) at the same rate, so that the total aperture length converges to a constant $2L$ (i.e. $\Delta_T M \rightarrow L$). As explained above, this asymptotic limit effectively models the general integral solution to the wave equation under spatially continuous excitation. 

Now, it can easily seen that as ($M\rightarrow\infty$) the eigenvalues of $\mathbf{H}^{t_\mathrm{pol} \times r_\mathrm{pol}} (\mathbf{H}^{t_\mathrm{pol} \times r_\mathrm{pol}})^H$ increase in magnitude, so that the overall capacity increases in magnitude as well. 
Therefore, we need to scale down the total transmit power as the number of transmitters grows without bound. To that effect, we will assume that we fix the total transmitted power according to
\begin{equation}
P=\frac{\bar{P}}{(2 M+1) t_{\mathrm{pol}}} 
\end{equation}
where $\bar{P}>0$ is a fixed constant. We will also define $\mathsf{SNR}_{RX}$ as the SNR at the receiver when the transmitter employs a matched filter (that is the principal right singular vector of  $\mathbf{H}^{t_\mathrm{pol} \times r_\mathrm{pol}}$). It can readily be seen that 
\[
\mathsf{SNR}_{RX} = \frac{\bar{P}}{\sigma^2 t_{\mathrm{pol} }}\left|\frac{\xi}{\lambda}\right|^2 s_M^{(2)}
\]
where $s_M^{(2)}$ is the inverse of the harmonic mean of the square distances between the receiver and the elements of the array,
\[
s_M^{(2)}=\frac{1}{2M+1} \sum_{m=-M}^{M} \frac{1}{\Vert\mathbf{r}_m\Vert^2}.
\]
It can be shown \cite{agustin24twc} that in the holographic regime this quantity converges to a fixed positive constant, $s_M^{(2)}\rightarrow \psi_2$, where\footnote{A number of similar quantities ($\psi_4,\psi_6,\bar{\psi}_3,\psi_4,\bar{\psi}_5$) are later introduced in the paper. For even $i=2,4,6$, the quantity $\psi_i$ can be seen to be the asymptotic arithmetic average of $\Vert \mathbf{r}_m \Vert^{-i}$ when $m=-M,\ldots,M$. For $i=3,5$, $\bar{\psi}_i$ is the asymptotic arithmetic average of $ (m\Delta_T -D\sin\theta ) \Vert \mathbf{r}_m \Vert^{-(i+1)}$ along $m=-M,\ldots,M$.}
\begin{equation}\label{eq:DefPsi2}
     \psi_2 = \frac{\Gamma}{2 L D \cos \theta}
\end{equation}
and where $\Gamma$ is the angle of view of the array from the receiver (see Fig.~\ref{fig:Scenario}), that is
\begin{equation}   \label{eq:defGamma}
    \Gamma = \arctan \frac{L/D-\sin\theta}{\cos\theta} + \arctan \frac{L/D + \sin\theta}{\cos{\theta}} . 
\end{equation}
This implies that $\mathsf{SNR}_{RX}\rightarrow \overline{\mathsf{SNR}}_{RX} = \bar{P}\vert \xi / \lambda \vert^2 \psi_2 /t_{\mathrm{pol}}/\sigma^2$.
Likewise, the capacity of the system also converges to a constant that can be described as follows. 
It was shown in \cite{agustin24twc} that 
\begin{equation} \label{eq:convergence}
\frac{1}{2M+1}
\mathbf{H}^{t_\mathrm{pol} \times r_\mathrm{pol}}(\mathbf{H}^{t_\mathrm{pol} \times r_\mathrm{pol}})^H
\rightarrow \left\vert\frac{\xi}{\lambda}\right\vert^2
\overline{\mathcal{W}}^{t_\mathrm{pol} \times r_\mathrm{pol}}
\end{equation}
where $\overline{\mathcal{W}}^{t_\mathrm{pol} \times r_\mathrm{pol}}$ is a certain Hermitian matrix that takes different values depending on the number of polarizations used at the transmitter and the receiver (the actual form of this matrix will be specfied later). Let $\gamma_i$, $i=1,\ldots, r_\mathrm{pol}$ denote the eigenvalues of this matrix, sorted in descending order. Then, it follows that the achievable spectral efficiency of this system converges to a fixed quantity
\begin{equation}
\bar{C}=\sum_{i=1}^{r_\mathrm{pol}} \log \left(1+\frac{\gamma_i}{\psi_2} {s}_i\right)
\end{equation}
where the coefficients ${s}_i>0$ are scaled versions of the power allocations that are fixed according to the waterfilling equation. More specifically, the optimum power coefficients are given by
\begin{equation}
{s}_i=\left[\frac{1}{\vartheta}-\frac{\psi_2}{\gamma_i}\right]^{+}
\end{equation}
where $[\cdot]^+ = \max(\cdot,0)$ and where $\vartheta$ is the waterlevel, which is obtained as the unique solution to the equation 
\begin{equation} \label{eq:waterlevel}
\overline{\mathsf{SNR}}_{RX}=\sum_{j=1}^{r_{\mathrm{pol}}} {s}_i=\sum_{i=1}^{r_{\mathrm{pol}}}\left[\frac{1}{\vartheta}-\frac{\psi_2}{\gamma_i}\right]^{+} .
\end{equation}
Now, we observe that the total number of streams that can be spatially multiplexed (achievable degrees of freedom) corresponds to the number of coefficients $s_i$ that are different from zero. We will denote by $n_+$ this number, which may take three different values ($n_+ = 1,2,3$) depending on the received SNR ($\overline{\mathsf{SNR}}_{RX}$) and the actual position of the receiver ($D,\theta$). 
In what follows, we provide a more detailed study of $n_+$ for different choices of $t_\mathrm{pol}$. More specifically, we will always assume that the receiver employs $r_\mathrm{pol}=3$ polarizations, whereas the transmitter uses either $t_\mathrm{pol} = 3$ or $t_\mathrm{pol} = 2$ polarizations (in this last case, aligned with the $x$ and $y$ axis). 

\subsection{Fully polarized transceiver ($t_\mathrm{pol}=r_\mathrm{pol}=3$)}
For the case of a fully polarized transmitter, the matrix $\overline{\mathcal{W}}^{3 \times 3}$ takes the form \cite{agustin24twc}
\[
\overline{\mathcal{W}}^{3\times3}= \left[
\begin{array}
[c]{ccc}%
\psi_{2} & 0 & 0\\
0 & \psi_{4}D^{2}\cos^{2}\theta & \bar{\psi}_{3}D\cos\theta\\
0 & \bar{\psi}_{3}D\cos\theta & \psi_{2}-\psi_{4}D^{2}\cos^{2}\theta
\end{array}
\right]
\]
where $\psi_2$ has been defined in (\ref{eq:DefPsi2}) and where 
\begin{align*}
\bar{\psi}_{3}  &  =-\frac{D\sin\theta}{\left(  D^{2}+L^{2}\right)  ^{2}-\left(
2LD\sin\theta\right)  ^{2}}\\
\psi_{4}  &  =\frac{1}{2}\frac{1}{\left(  D\cos\theta\right)  ^{2}}\left[
\frac{\left(  D^{2}+L^{2}\right)  -2D^{2}\sin^{2}\theta}{\left(  D^{2}
+L^{2}\right)  ^{2}-\left(  2LD\sin\theta\right)  ^{2}}+ \psi_{2}\right].
\end{align*}
In order to analyze the behavior of the number of active spatial degrees of freedom $n_+$ we need to study the three eigenvalues of the matrix $\overline{\mathcal{W}}^{3\times3}$. They can be provided in closed form, as specified in the following result.
\begin{lemma}
    The three eigenvalues of $\overline{\mathcal{W}}^{3\times3}$  are given by $\gamma_1 = \psi_2$, $\gamma_2 = (\psi_2+\Delta^{-1})/2$ and $\gamma_3 = (\psi_2-\Delta^{-1})/2$, where
    $$
    \Delta = \sqrt{\left(D^2-L^2\right)^2+(2 L D \cos \theta)^2}.
    $$
    Furthermore we always have $\gamma_1>\gamma_2>\gamma_3>0$. 
\end{lemma}
\begin{proof}
    The expression of the three eigenvalues can be trivially derived by finding the roots of the corresponding characteristic equation. In order to show that the ordering is correct and $\gamma_3$ is positive, we need to prove that $\Delta^{-1}<\psi_2$ or, equivalently, $\psi_2\Delta>1$. Observe that, using the definition of $\psi_2$ in (\ref{eq:DefPsi2}), we can express this quantity as
    \[
    \psi_{2}\Delta=\Gamma\sqrt{1+\left(  \frac{D^{2}-L^{2}}{2LD \cos \theta
}\right)  ^{2}}
    \]
    where $\Gamma$ has been defined in (\ref{eq:defGamma}). 
    Now, if $L\geq D$ then $\Gamma \geq\pi/2$ (see Fig.~\ref{fig:Scenario}) and the inequality is obvious because the term inside the square root is larger than or equal to $1$. If $L<D$ we can use the formula for the sum of arc-tangents and the fact that $\arctan(x)>x/\sqrt{1+x^2}$ for $x>0$ so that
\[
\Gamma=\arctan
\frac{2LD\cos\theta}{D^{2}-L^{2}}>\left(1+\left(  \frac{D^{2}-L^{2}
}{2LD \cos\theta} \right)  ^{2}\right)^{-1/2}
\]
from where $\psi_2\Delta>1$ directly follows.
\end{proof}

Based on the above description of the eigenvalues, we have now all the ingredients to characterize the minimum SNR that is required for the activation of the different spatial degrees of freedom. This can be done by noticing that the right hand side of (\ref{eq:waterlevel}) is an increasing function of $1/\vartheta$ that consists of the sum of three linear terms with intersection at the points $\vartheta^{-1} = \psi_2/\gamma_i$, $i=2,3$. Hence, the value on the right hand side of (\ref{eq:waterlevel}) at these two points will establish the two thresholds in $\overline{\mathsf{SNR}}_{RX}$ that activate two or three spatial streams respectively. A direct evaluation shows that the two thresholds take the form
\begin{align} 
    \overline{\mathsf{SNR}}^{(1)} = & \frac{\psi_2}{\gamma_2}-\frac{\psi_2}{\gamma_1} = \frac{\psi_2\Delta-1}{\psi_2\Delta + 1} \label{eq:SNRthresholds3x3}  \\ 
    \overline{\mathsf{SNR}}^{(2)} = & \frac{2\psi_2}{\gamma_3}-\frac{\psi_2}{\gamma_1} - \frac{\psi_2}{\gamma_2} = \frac{\psi_2\Delta-1}{\psi_2\Delta + 1}+ \frac{8\psi_2\Delta}{\left(\psi_2\Delta\right)^2 - 1}.\label{eq:SNRthresholds3x32}
\end{align}
We can conclude that the optimum number of active streams $n_+$ is equal to 
\begin{equation} \label{eq:activationModes}
n_+ = \left\{
\begin{array}
[c]{cc}
1 & \overline{\mathsf{SNR}}_{RX} < \overline{\mathsf{SNR}}^{(1)} \\
2 &  \overline{\mathsf{SNR}}^{(1)} \leq \overline{\mathsf{SNR}}_{RX} < \overline{\mathsf{SNR}}^{(2)}\\
3 &  \overline{\mathsf{SNR}}^{(2)}\leq \overline{\mathsf{SNR}}_{RX} .
\end{array}
\right. 
\end{equation}
In the next section we will establish how this translates into different spatial regions where $n_+ =1$, $n_+ =2$ or $n_+=3$ spatial streams can be optimally transmitted. 

\subsection{Double polarized transmitter ($t_\mathrm{pol}=2, r_\mathrm{pol}=3$)}
Consider there the case where the transmitter only employs two polarizations, which are aligned with the $x$ and $y$ axis respectively. In this case, one can show \cite{agustin24twc} that the convergence in (\ref{eq:convergence}) holds with
\begin{equation}
\overline{\mathcal{W}}^{2 \times 3}=\left[\begin{array}{ccc}
\psi_2 & 0 & 0 \\
0 & \psi_6 D_c^4 & \bar{\psi}_5 D_c^3 \\
0 & \bar{\psi}_5 D_c^3 & \psi_4 D_c^2 -\psi_6 D_c^4
\end{array}\right] 
\end{equation}
where we have introduced the short-hand notation $D_s = D \sin \theta$, $D_c = D\cos \theta$ together with 
\begin{align*}
\bar{\psi}_{5}  &  =-\frac{\left(  D^{2}+L^{2}\right)  D_s}{\left(  \left(
D^{2}+L^{2}\right)  ^{2}-\left(  2LD_s \right)  ^{2}\right)  ^{2}}\\
\psi_{6}  &  = \frac{1}{4D_c^{2}}
\frac{\left(  D^{2}+L^{2}\right)  ^{2}-4 D^2 D_s^2}{\left(  \left(
D^{2}+L^{2}\right)  ^{2}-\left(  2LD_s\right)  ^{2}\right)  ^{2}
} +\frac{3}{D_c^2}\psi_4.
\end{align*}
As in the previous subsection, we provide a closed form expression for the eigenvalues of this matrix in order to analyze the activation of the different spatial streams. 
\begin{lemma}
       The three eigenvalues of $\overline{\mathcal{W}}^{2\times3}$  are given by $\gamma_1 = \psi_2$, 
       \begin{align}
           \gamma_2 & =\frac{D_c^2}{2}\left[\psi_4+ \sqrt{\left[\psi_4-2 D_c^2 \psi_6\right]^2+4 D_c^2 \bar{\psi}_5^2}\right] \\
            \gamma_3 & =\frac{D_c^2}{2}\left[\psi_4 - \sqrt{\left[\psi_4-2 D_c^2 \psi_6\right]^2+4 D_c^2 \bar{\psi}_5^2}\right].
       \end{align}
    Furthermore we always have $\gamma_1>\gamma_2>\gamma_3>0$. 
\end{lemma}
\begin{proof}
As before, the expression for the eigenvalues follows directly from finding the roots of the corresponding characteristic equation. It remains to establish that the ordering holds and that $\gamma_3>0$. We claim that $D_c^2\psi_4 < \psi_2$. Indeed, it can be shown \cite{agustin24twc} that $\psi_4$ can also be expressed as the integral
\begin{multline}
D_c^2 \psi_4=\frac{1}{2 L} \int_{-L}^L \frac{D_c^2}{\left((x-D_s)^2+D_c^2\right)^2} dx \\
\leq \frac{1}{2 L} \int_{-L}^L \frac{1}{\left((x-D_s)^2+ D_c^2\right)} d x=\psi_2
\end{multline}
where the inequality follows from the fact that $D_c^2\leq(x-D_s)^2+D_c^2$. Having shown this, the result will be proven if we are able to establish
\begin{equation}
\frac{1}{2} \sqrt{\left[\psi_4-2 D_c^2 \psi_6\right]^2+4 D_c^2 \bar{\psi}_5^2}<\psi_4
\end{equation}
or, equivalently
\begin{equation} \label{eq:ineq2prove}
D_c^2 \bar{\psi}_5^2<\frac{3}{4} \psi_4^2+\left[\psi_4-D_c^2 \psi_6\right]D_c^2 \psi_6.
\end{equation}
To see this, we can use the integral definitions of $\psi_4,\bar{\psi}_5,\psi_6$ in \cite{agustin24twc} and consider the matrix
\begin{equation*}
    \left[\begin{array}{cc}
D_c^2 \psi_6 & D_c \bar{\psi}_5 \\
D_c \bar{\psi}_5 & \psi_4- D_c^2 \psi_6
\end{array}\right]= \frac{1}{2 L} \int_{-L}^L \frac{\mathbf{r}(x) \mathbf{r}^H(x)}{\Vert \mathbf{r}(x) \Vert^6} dx
\end{equation*}
 where $\mathbf{r}(x) = [D_c,x-D_s]^T$. 
This matrix is positive definite by definition, so that its determinant is positive, implying that $D_c^2 \psi_6(\psi_4-D_c^2 \psi_6)>D_c^2 \bar{\psi}_5^2 $, from where (\ref{eq:ineq2prove}) follows. 
\end{proof}
Using the above result and following the same approach as in the previous subsection, we can derive a closed form expression for the threshold SNRs that guarantee the activation of the different spatial modes. Hence, one can easily establish the activation rules in (\ref{eq:activationModes}), where now the thresholds $\overline{\mathsf{SNR}}^{(1)}$ and $\overline{\mathsf{SNR}}^{(2)}$ take the expression in (\ref{eq:snrth2x31}) and (\ref{eq:snrth2x32}) as shown at the top of the next page. 
\begin{figure*}[t]
    \centering
    \normalsize
\begin{align}
\overline{\mathsf{SNR}}^{(1)} &  =\frac{D^{2}\psi_{2}\left[  D^{4}\psi
_{4}-\sqrt{\left[  D^{4}\psi_{4}-2\cos^{2}\theta\left(  D^{6}\psi_{6}\right)
\right]  ^{2}+4\cos^{2}\theta\left(  D^{5}\psi_{5}\right)  ^{2}}\right]
}{2\cos^{4}\theta\left[  D^{6}\psi_{6}\left(  D^{4}\psi_{4}-\cos^{2}\theta
D^{6}\psi_{6}\right)  -\left(  D^{5}\psi_{5}\right)  ^{2}\right]  }-1 \label{eq:snrth2x31}\\
\overline{\mathsf{SNR}}^{(2)} &  =\frac{D^{2}\psi_{2}\left[  D^{4}\psi
_{4}+3\sqrt{\left[  D^{4}\psi_{4}-2\cos^{2}\theta\left(  D^{6}\psi_{6}\right)
\right]  ^{2}+4\cos^{2}\theta\left(  D^{5}\psi_{5}\right)  ^{2}}\right]
}{2\cos^{4}\theta\left[  D^{6}\psi_{6}\left(  D^{4}\psi_{4}-\cos^{2}\theta
D^{6}\psi_{6}\right)  -\left(  D^{5}\psi_{5}\right)  ^{2}\right]  }-1 \label{eq:snrth2x32}
\end{align}
    \hrulefill
\end{figure*}
\begin{remark} \label{remark:dipsii}
    It can easily be established using the expressions of the quantities $\psi_i,\bar{\psi}_i$ that both $D^i \psi_i$ and $D^i \bar{\psi}_i$ are functions of the elevation $\theta$ and the ratio $D/L$. This implies that the corresponding SNR thresholds (as given by (\ref{eq:SNRthresholds3x3})-(\ref{eq:SNRthresholds3x32}) in the $3\times 3$ polarization case and by (\ref{eq:snrth2x31})-(\ref{eq:snrth2x32}) in the $2\times 3$ polarization case) uniquely depend on these two variables.
\end{remark}
The above conditions establish the practical availability of the different spatial modes depending on the SNR at the receiver when the transmitter is employing a spatial matched filter (beamfocusing) towards the receiver, that is $\overline{\mathsf{SNR}}_{RX}$. However, the value of $\overline{\mathsf{SNR}}_{RX}$ will decay with the distance so it is unclear which will be the spatial regions where the different number of streams will be optimally available. In the next section we focus on the characterization of these regions. 

\section{Support regions of spatial multiplexing}
In this section we study the effective coverage areas or multiplexing regions where $n_+=1,2,3$ spatial streams are available. 
Before going into the details, it should be pointed out that this study is essentially different from previous approaches that investigated the \emph{total} number of available spatial degrees of freedom of a holographic surface in various contexts \cite{Hu18,Dardari20,Pizzo22b,Ji23}. The focus here is on characterizing the regions where these spatial degrees of freedom are effectively achievable.
More specifically, in order to establish the regions where $n_+ =1,2,3$ spatial streams are available, we will fix the received SNR at a certain reference point in space that is located $L$ meters away along the broadside of the array. (This point has coordinates $(0,0,L)$ in the reference system of Fig.~\ref{fig:Scenario}.) Let us denote as $\overline{\mathsf{SNR}}_0$ the value of the received SNR at this point. Based on this, we can express the received SNR at the intended receiver as
\begin{equation} \label{eq:rel_SNR0_SNRrx}
\overline{\mathsf{SNR}}_{RX} = \overline{\mathsf{SNR}}_0 \frac{4L^2\psi_2}{\pi} = \overline{\mathsf{SNR}}_0 \frac{4(D^2\psi_2)}{(D/L)^2\pi} 
\end{equation}
which can be easily seen by observing from the expression of $\psi_2$ in (\ref{eq:DefPsi2}) that $\psi_2 = \pi / (4L^2)$  when $\theta = 0$ and $D=L$ (as it is the case in the reference point). By Remark \ref{remark:dipsii} above, the right hand side of (\ref{eq:rel_SNR0_SNRrx}) is a function of $D/L$ and $\theta$ only.

Now, let us recall from the previous section that one can establish the different activation conditions by comparing $\overline{\mathsf{SNR}}_{RX} $ against $\overline{\mathsf{SNR}}^{(1)} (\theta,D/L) $ and $\overline{\mathsf{SNR}}^{(2)} (\theta,D/L) $, where we made explicit the dependence of the two thresholds as functions of the elevation $\theta$ and the ratio $D/L$, see the activation conditions in (\ref{eq:activationModes}). Now, replacing the expression of $\overline{\mathsf{SNR}}_{RX}$ in (\ref{eq:activationModes}) with the right hand side of (\ref{eq:rel_SNR0_SNRrx}) above one can establish equivalent activation conditions on the pair $(D/L,\theta)$ that determine the activation of one, two or three spatial modes respectively. Indeed, the two boundaries can be determined by solving the equations 
\begin{equation} \label{eq:normalization}
\overline{\mathsf{SNR}}_0 = \frac{\pi}{4D^2\psi_2} \left(\frac{D}{L}\right)^2 \overline{\mathsf{SNR}}^{(i)}(\theta,D/L) 
\end{equation}
for $i=1,2$. The expression of $\overline{\mathsf{SNR}}^{(i)}(\theta,D/L)$ is different depending on the number of polarizations used, and has been obtained in (\ref{eq:SNRthresholds3x3})-(\ref{eq:SNRthresholds3x32}) and (\ref{eq:snrth2x31})-(\ref{eq:snrth2x32}) for the cased of $t_\mathrm{pol}=3$ and $t_\mathrm{pol}=2$ respectively. 

With some abuse of notation, let us denote the right hand side of (\ref{eq:normalization}) as $\overline{\mathsf{SNR}}^{(i)}_0(\theta,D/L)$, so that the boundary condition in (\ref{eq:normalization}) can be equivalently formulated as 
\begin{equation} \label{eq:normalized2}
\overline{\mathsf{SNR}}_0=\overline{\mathsf{SNR}}^{(i)}_0(\theta,D/L).
\end{equation} Fig.~\ref{fig:SNRthr3x3} and Fig.~\ref{fig:SNRthr2x3} provide a representation of $\overline{\mathsf{SNR}}_0(\theta,D/L)$ as a function of $D/L$ for the cases $t_\mathrm{pol}=r_\mathrm{pol}=3$ (Fig.~\ref{fig:SNRthr3x3}) and  $t_\mathrm{pol}=2, r_\mathrm{pol}=3$ (Fig.~\ref{fig:SNRthr2x3}) respectively.  We can graphically represent the conditions in (\ref{eq:activationModes}) by considering a horizontal line in all these plots at the value of the SNR in the reference point ($\overline{\mathsf{SNR}}_0$). The crossing points between this horizontal line and the different curves will correspond to the values of $D/L$ at the boundary of the activation regions. In general terms, we observe that $\overline{\mathsf{SNR}}_0^{(i)}(\theta,D/L)$ is an increasing function of $D/L$, indicating that there exists a single value, denoted as $(D/L)^{(i)}_\mathrm{th}$, that determines the boundary. Note, however, that there exist some specific cases where this is not the case and the corresponding activation regions will be slightly different. 

\begin{figure}[t]
    \centering
    \includegraphics[width=\columnwidth]{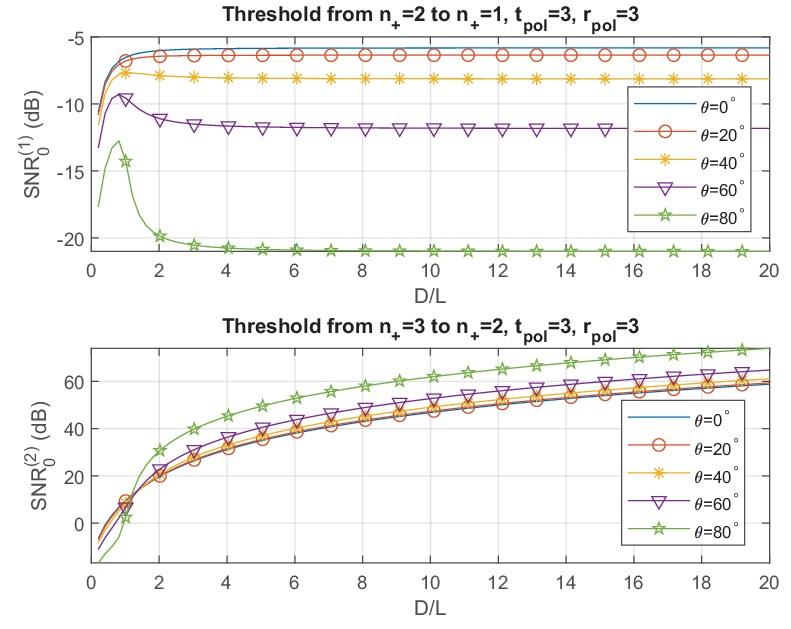}
    \caption{Representation of the right hand side of (\ref{eq:normalization}) as a function of $D/L$ for fixed elevations $\theta$ when $t_\mathrm{pol}=r_\mathrm{pol}=3$.}
    \label{fig:SNRthr3x3}
\end{figure}
\begin{figure}[t]
    \centering
    \includegraphics[width=\columnwidth]{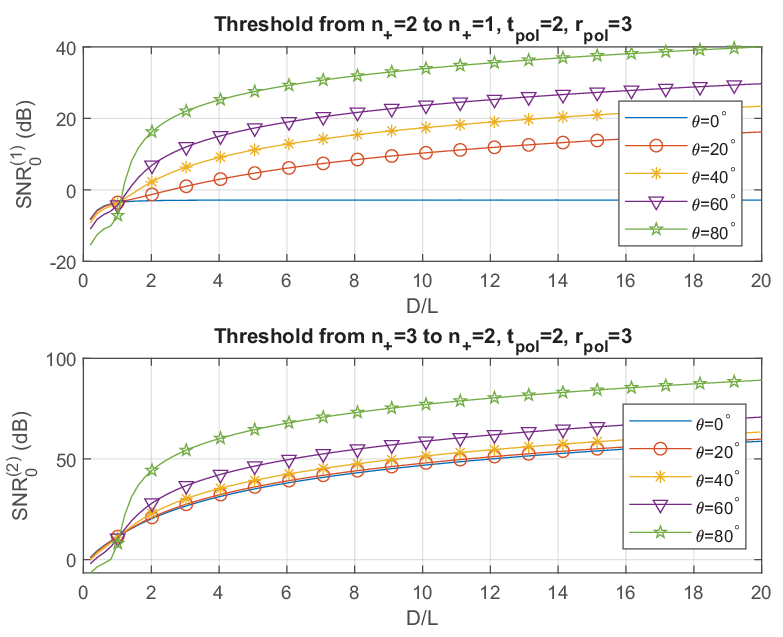}
    \caption{Representation of the right hand side of (\ref{eq:normalization}) as a function of $D/L$ for fixed elevations $\theta$ when $t_\mathrm{pol}=2, r_\mathrm{pol}=3$.}
    \label{fig:SNRthr2x3}
\end{figure}

More specifically, observing the upper plot in Fig.~\ref{fig:SNRthr3x3} one can readily see that for any fixed $\theta$, the right hand side of (\ref{eq:normalization}) converges to a constant as $D/L \rightarrow \infty$, i.e. 
\begin{equation} \label{eq:limit3x2boundary1}
\lim_{D/L\rightarrow \infty} \overline{\mathsf{SNR}}_0^{(1)}(\theta,D/L) = \frac{\pi}{12}\cos^2(\theta).
\end{equation}
This implies that if the SNR at the reference point is such that $\overline{\mathsf{SNR}}_0> \pi/12$, the transmission will always achieve $n_+=2$ spatial degrees of freedom no matter how far away the receiver is. A very similar behavior is observed in the upper plot of Fig.~\ref{fig:SNRthr2x3} when $\theta = 0$ and $t_\mathrm{pol} =2$. In this case, it can readily be seen that 
\begin{equation} \label{eq:limit3x3boundary1}
\lim_{D/L\rightarrow \infty} \overline{\mathsf{SNR}}_0^{(1)}(0,D/L) =\frac{\pi}{6}.
\end{equation}
Hence, if the SNR at the reference location is above this value, $n_+=2$ spatial degrees of freedom will always be available along the broadside of the array. For all the other situations where $\overline{\mathsf{SNR}}_0^{(i)}(\theta,D/L)$ is a monotonically increasing function of $D/L$, a higher number of spatial streams will be available as the receiver gets closer to the array. 

The thresholds obtained using the asymptotic holographic approximation can be seen to approximate quite accurately the real thresholds obtained with the actual channel response for fixed values of $M$. These real thresholds can be obtained by following the above procedure, but using the eigenvalues of the true channel matrix (for fixed $M$) instead of their asymptotic approximation under the holographic regime. Fig.~\ref{fig:errorThresholds} represents
the error between the true activation distances (finite $M$) and the asymptotic ones (when $M \rightarrow \infty$) as a function of increasing $M$. The one-sided array size is kept fixed to $L=1$ meter so that the inter-element separation $\Delta_T = 1/M$ decreases as $M$ becomes larger. One can see from this figure that the relative error incurred by the holographic approximation is quite small even for relatively low values of $M$. For example, the relative error in the determination of the boundary is lower than $6.4$\% when $M=4,\Delta_T = 25$cm and lower than $1.8$\% when $M=16,\Delta_T = 6.25$cm. The holographic approximation can therefore be safely used in more conventional situations where the inter-element separation of the ULA is comparable in magnitude with the transmission wavelength. 

\begin{figure}
    \centering
    \includegraphics[width=\columnwidth]{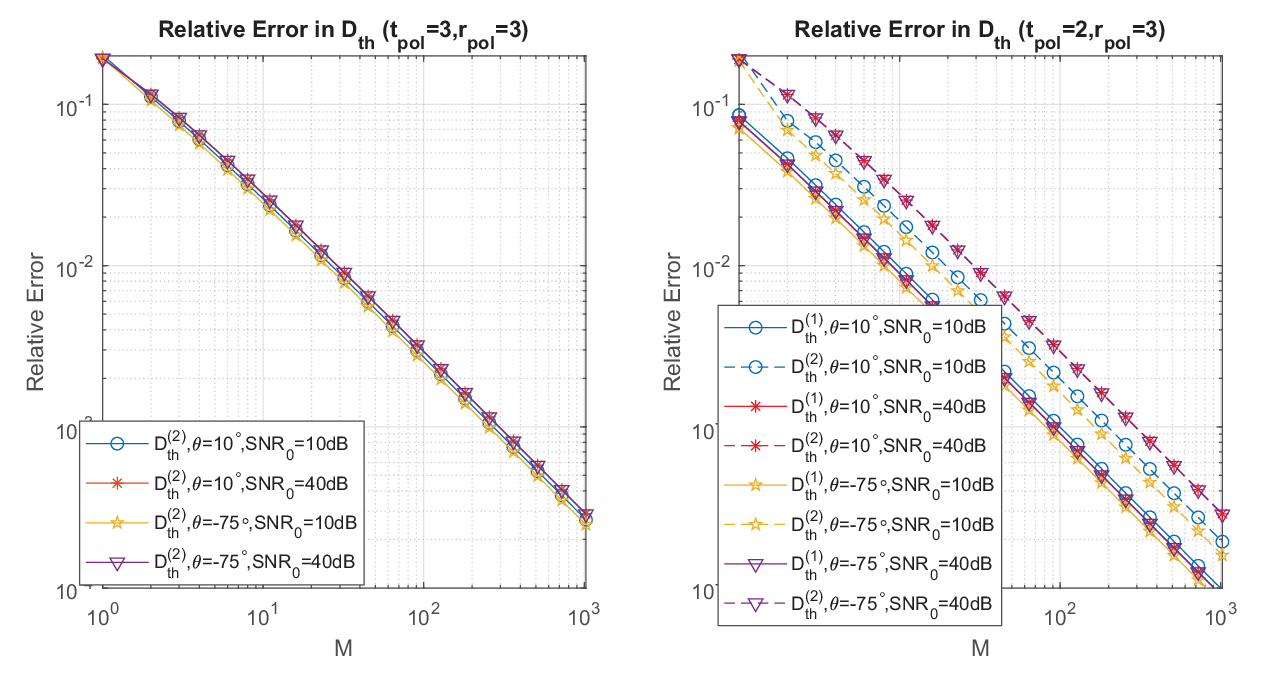}
    \caption{Relative error between the true distance activation boundaries and the holographic ones (denoted as $D_{th}^{(i)}$ for $i=1,2$) as a function of $M$ for fixed $L=\Delta_T M =1$ and some selected values of $\theta$ and $\overline{\mathsf{SNR}}_0$.}
    \label{fig:errorThresholds}
\end{figure}

\subsection{Simplified expressions for the multiplexing boundaries}

Given the complicated analytical expressions of the thresholds in (\ref{eq:SNRthresholds3x3})-(\ref{eq:SNRthresholds3x32}) and (\ref{eq:snrth2x31})-(\ref{eq:snrth2x32}), it becomes difficult to provide further insights into the geometrical description of the different multiplexing regions. We now try to shed some light into the shape of these regions by considering the approximation $D/L>>1$.
\begin{lemma} \label{lem:highD3x3}
    For any fixed $\theta$, the function $\overline{\mathsf{SNR}}_0^{(i)}(\theta,D/L)$ on the right hand side of (\ref{eq:normalized2}) when $t_\mathrm{pol}=r_\mathrm{pol}=3$ behaves for $D/L\rightarrow\infty$ as
\begin{align*}
\overline{\mathsf{SNR}}_{0}^{(1)} &  =\frac{\pi}{12}\cos^{2}\theta+O((L/D)^2)  \\
\overline{\mathsf{SNR}}_{0}^{(2)} &  =\frac{3\pi}{2\cos^{2}\theta}\left[
\left(  \frac{D}{L}\right)  ^{4}-2A_{1}\left(  \theta\right)  \left(  \frac
{D}{L}\right)  ^{2}+A_{2}\left(  \theta\right)  \right]  \\ &+O((L/D)^2)
\end{align*}
where $A_1(\theta) = (3-58/15\cos^2(\theta))/2$ and 
\[
 A_{2}\left(  \theta\right)    =2\left(  1-\frac{4}{3}\cos^{2}\theta\right).
\]
  
\end{lemma}
\begin{proof}
    See the Appendix.
\end{proof}
By disregarding the terms that converge to zero in the above approximation, one may find a reasonably accurate description of the different activation areas. For example, the condition for the activation of at least two modes ($n_+>1$) everywhere can be approximately guaranteed if $\overline{\mathsf{SNR}}_{0}>\pi/12 = -5.82\mathrm{dB}$. When this condition holds, we have $(D/L)^{(1)}_\mathrm{th}=+\infty$ and the optimum transmitter employs at least two streams everywhere for reasonable values of the reference SNRs. The distance that separates the region between $n_+=2$ and $n_+=3$ active spatial modes, on the other hand, takes the form 
\[
\left(  D/L\right)  _{\mathrm{th}}^{(2)}=\sqrt{A_{1}\left(  \theta\right)
+\sqrt{A_{1}^{2}\left(  \theta\right)  -A_{2}\left(  \theta\right)
+\frac{2\cos^{2}\theta}{3\pi}\overline{\mathsf{SNR}}_{0}}}
\]
which essentially scales up as the fourth root of the SNR at the reference point.
A similar asymptotic analysis is possible when $t_\mathrm{pol}=2,r_\mathrm{pol}=3$. The following lemma provides a large $D/L$ approximation of the holographic boundaries between multiplexing regions. 
\begin{lemma} \label{lem:highD2x3}
    For any fixed $\theta$, the function $\overline{\mathsf{SNR}}_0^{(i)}(\theta,D/L)$ on the right hand side of (\ref{eq:normalized2}) when $t_\mathrm{pol}=2, r_\mathrm{pol}=3$ behaves for $D/L\rightarrow\infty$ as
\begin{align*}
\overline{\mathsf{SNR}}_{0}^{(1)}  & =\frac{\pi}{4}\left[  \left(  \frac{D}%
{L}\right)  ^{2}\tan^{2}\theta-\frac{4\cos^{4}\theta-16\cos^{2}\theta
+10}{3\cos^{2}\theta}\right] \\ &+O((L/D)^2) 
\\
\overline{\mathsf{SNR}}_{0}^{(2)}  & =\frac{3\pi}{2\cos^{4}\theta}\left[
\left(  \frac{D}{L}\right)  ^{4}-2B_{1}(\theta)\left(  \frac{D}{L}\right)
^{2}+B_{2}(\theta)\right]   \\ &+O((L/D)^2) 
\end{align*}
where 
\begin{align*}
B_{1}(\theta)  & =\frac{5\cos^{4}\theta-177\cos^{2}\theta+128}{60}\\
B_{2}(\theta)  & =\frac{-350\cos^{6}\theta + 10242\cos^{4}\theta-19311\cos
^{2}\theta+10122}{1575}.
\end{align*}
\end{lemma}
\begin{proof}
    See the Appendix.
\end{proof}
Here again, one can approximate the boundaries of the activation regions by disregarding the terms that converge to zero in the above expressions. Special care is needed when analyzing the threshold $\overline{\mathsf{SNR}}_{0}^{(1)}$. When $\theta=0$ this quantity converges to a constant $\pi/6$ as $D/L\rightarrow \infty$. This means that at least $n_+ = 2$ modes will always be available along the broadside of the ULA\footnote{In fact, it can readily be seen from the above expansion that at least two spatial modes will be achievable at sufficiently large distances provided that the receiver is located on the strip $\vert y_0\vert ^2 < 4\overline{\mathsf{SNR}}_0/\pi-2/3$. \label{footnote:strip}}.  
For the rest of the cases, we can approximate the two boundary distances as 
\begin{align}
\left(  D/L\right)  _{\mathrm{th}}^{(1)} &  =\sqrt{\frac{4\cos^{4}\theta
-16\cos^{2}\theta+10}{3\sin^{2}\theta}+\frac{4\overline{\mathsf{SNR}}_{0}}{\pi\tan^{2}\theta}} \label{eq:boundary1_2x3}\\
\left(  D/L\right)  _{\mathrm{th}}^{(2)} &  = \sqrt{ B_{1}(\theta)+\sqrt{B_{1}%
^{2}(\theta)-B_{2}(\theta)+\frac{2\cos^{4}\theta}{3\pi}\overline{\mathsf{SNR}
}_{0}}} \nonumber
\end{align}
where we notice that $(D/L)_{\mathrm{th}}^{(1)}$ scales up as the square root of the reference $\overline{\mathsf{SNR}}_{0}$ whereas  $(D/L)_{\mathrm{th}}^{(2)}$ scales up more slowly, as the fourth root of the same quantity. 

\subsection{Geometric representation and discussion}

Fig.~\ref{fig:boundary1} represents the boundary between the region where $n_+=2$ and the region where $n_+=1$ for different values of the reference $\overline{\mathsf{SNR}}_0$. The same figure represents in solid lines the actual asymptotic boundary that is obtained by numerically inverting the equation in (\ref{eq:normalization}) and in dash-dotted lines the approximation obtained by assuming large $D/L$ in Lemma~\ref{lem:highD2x3} (only the case $t_\mathrm{pol}=2$ is represented, since there is no boundary for reasonable SNRs when $t_\mathrm{pol}=3$). Observe that the approximation is extremely accurate for reasonable values of the reference $\overline{\mathsf{SNR}}_0$. Hence, one can safely approximate the region where three spatial streams are optimally available by the region where (\ref{eq:boundary1_2x3}) holds true. On the other hand, Fig.~\ref{fig:boundary2} represents the activation boundary between $n_+=2$ and $n_+=3$ spatial streams when $t_\mathrm{pol}=3$ (red lines) and $t_\mathrm{pol}=2$ (blue lines) for different values of the reference $\overline{\mathsf{SNR}}_0$. Both the true holographic regime (solid lines) and the high-$D/L$ approximations in Lemmas \ref{lem:highD3x3}-\ref{lem:highD2x3} (dash-dotted lines) are represented in the figure. Observe here that the activation regions for $t_\mathrm{pol}=3$ and $t_\mathrm{pol}=2$ are very similar regardless of $\overline{\mathsf{SNR}}_0$, although a slightly larger region where $n_+=3$ is obtained when $t_\mathrm{pol}=3$. 
\begin{figure}[t]
    \centering
    \includegraphics[width=\linewidth]{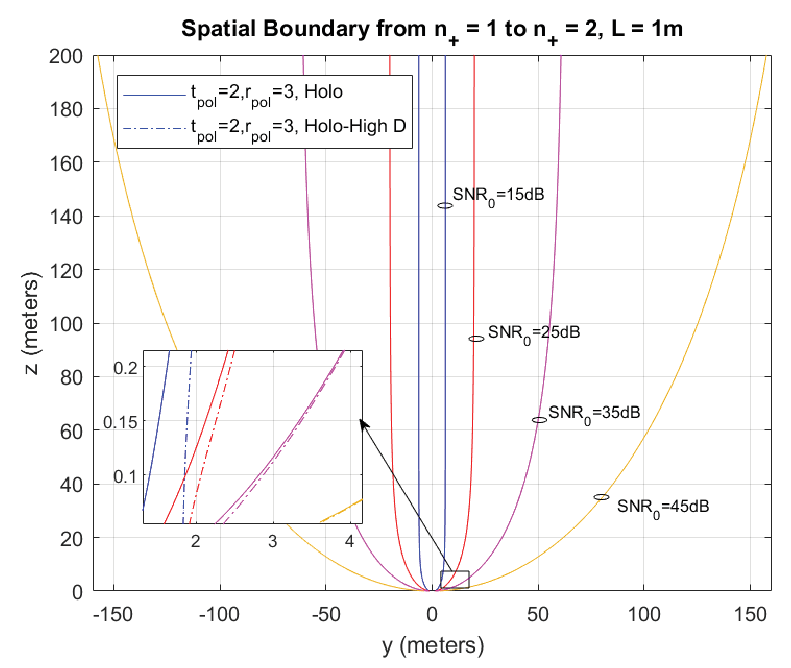}
    \caption{Boundary between $n_+=2$ and $n_+=1$ when $t_\mathrm{pol}=2, r_\mathrm{pol}=3$ for different values of $\overline{\mathsf{SNR}}_0$. Solid lines represent the actual asymptotic boundaries and dash-dotted lines represent the high $D/L$ approximations.}
    \label{fig:boundary1}
\end{figure}
\begin{figure}[t]
    \centering
    \includegraphics[width=\linewidth]{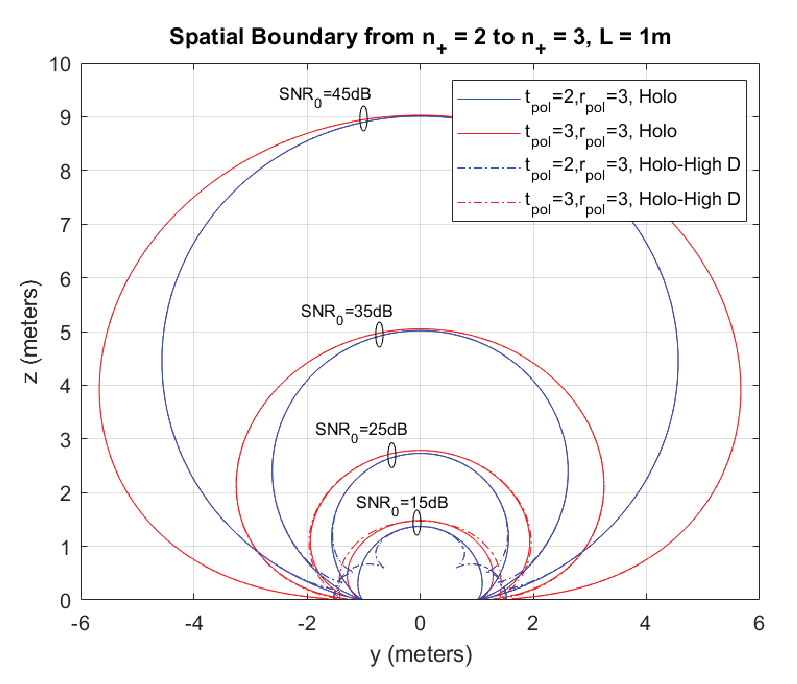}
    \caption{Boundary between the regions where $n_+=2$ and $n_+=3$ for different values of the reference $\overline{\mathsf{SNR}}_0$. Solid lines represent the actual asymptotic boundaries and dash-dotted lines represent the high $D/L$ approximations.}
    \label{fig:boundary2}
\end{figure}

From all this we can conclude that the use of three orthogonal polarizations guarantees that at least two spatial streams can be supported everywhere, provided that the transmit power is reasonably high so that the achievable SNR at a point located $L$ meters along the broadside of the array is higher than $10\log_{10}(\pi/12)=-5.82\mathrm{dB}$, cf. (\ref{eq:limit3x2boundary1}). This is in stark contrast with the setting where only two orthogonal polarizations are used, which can only support one stream in distances outside a certain vertical strip (see footnote~\ref{footnote:strip}), provided that $\overline{\mathsf{SNR}}_0>\pi/6$ according to (\ref{eq:limit3x3boundary1}). It may be argued that the two configurations are not really comparable, since the number of radiating elements is higher when $t_\mathrm{pol}=3$. This is indeed true, although the conclusions do not differ significantly when the number of radiating elements is the same in the two configurations. Note that increasing the number of elements of the ULA is equivalent to increasing $L$ in the above analysis and, since all the relevant quantities depend on the quotient $D/L$, the same conclusions hold by proportionally scaling the coverage distances $D$. 


\section{Conclusions}
A detailed study of the multi-stream transmit capabilities of a ULA towards a single element receiver (both employing multiple polarizations) has been provided. The analysis is based on a holographic approximation assuming that the radiating elements are infinitesimal dipoles. The number of elements of the ULA is assumed to be asymptotically large while the separation between consecutive elements becomes proportionally small, so that the total aperture length converges to a constant. Among other insights, the study allows to conclude that the use of three orthogonal polarizations at the transmitter guarantees that at least two spatial streams are available everywhere under minimal transmit power conditions. 

\appendices
\section{Sketch of the Proof of Lemmas \ref{lem:highD3x3} and \ref{lem:highD2x3}}
The proof is based on the analysis of the quantities $D^i \psi_i$ for $i=2,\ldots,6$ as 
$L/D  \rightarrow0$. For any fixed $\theta$, all these quantities can be seen as analytic functions of the ratio $L/D$, so one can establish their Taylor expansion around $L/D=0$, namely 
\begin{align*}
D^{2}\psi_{2} &  =1+\left(  1-4/3\cos^{2}\theta\right) (L/D)^{2}+O((L/D)^4)  \\
D^{4}\psi_{4} & = 1+2\left(  5/3-2\cos^{2}\theta\right)  (L/D)^{2} +O((L/D)^4)  \\
D^{5}\bar{\psi}_{5} &  = -\sin\theta - \left(  5-8\cos^{2}\theta\right)\sin\theta (L/D)^{2}  +O((L/D)^4)  \\
D^{6}\psi_{6} &  =1+\left(
7-8\cos^{2}\theta\right)  (L/D)^{2} +O((L/D)^4).
\end{align*}
In the same way, one can also establish that 
\begin{multline*}
    \psi_2\Delta  =1+ 2/3\cos^{2}\theta (L/D)^{2}+ \\
+2/3\cos^{2}\theta\left(  2- {11}/{5}\cos^{2}\theta\right) (L/D) ^{4}+O((L/D)^{6}).
\end{multline*}
A direct application of these expansions in (\ref{eq:SNRthresholds3x3})-(\ref{eq:SNRthresholds3x32}) and (\ref{eq:snrth2x31})-(\ref{eq:snrth2x32}) allows to establish the results in these two lemmas (details are omitted due to space constraints). 



\ifCLASSOPTIONcaptionsoff
  \newpage
\fi



\bibliographystyle{IEEEtran}
\bibliography{IEEEabrv,bibtex/bib/ICASSP25}

\end{document}